\documentclass[12pt]{iopart}
\usepackage{hyperref}
\usepackage{amsthm}
\usepackage{dsfont}
\usepackage{bm}
\newtheorem{theorem}{Theorem}
\newtheorem{definition}{Definition}
\newtheorem{game}{Game}
\newtheorem{task}{Task}
\newtheorem{lemma}{Lemma}
\newtheorem{example}{Example}
\usepackage{tikz}
\usetikzlibrary{positioning}
\begin{document}

\def\tit{Device-independent test of causal order and relations to fixed-points}
\title{\tit}

\author{\"Amin Baumeler$^{1,2}$ and Stefan Wolf$^{1,2}$}

\address{$^1$ Faculty of Informatics, Universit\'{a} della Svizzera italiana, Via G. Buffi 13, 6900 Lugano, Switzerland}
\address{$^2$ Facolt\`{a} indipendente di Gandria, Lunga scala, 6978 Gandria, Switzerland}
\ead{baumea@usi.ch}

\begin{abstract}
	Bell non-local correlations cannot be naturally explained in a
	fixed causal structure. This serves as a motivation for considering
	models where no global assumption is made beyond logical
	consistency. The assumption of a fixed causal
	order between a set of parties, together with free
	randomness, implies device-independent inequalities --- just as the assumption of locality does.
	It is known that local validity of quantum theory is consistent
	with violating such inequalities. Moreover, for three parties or more,
	even the (stronger) assumption of local {\em classical\/} probability theory plus logical
	consistency allows for violating causal inequalities.
	Here, we show that a classical environment (with which the parties
	interact), possibly containing loops, is logically consistent if and only
	if whatever the involved parties do, there is exactly one fixed-point,
	the latter being representable as a mixture of deterministic fixed-points.
	We further show that the non-causal view allows for a model
	of computation strictly more powerful than computation in a
	world of fixed causal orders.
\end{abstract}
\tableofcontents
\def\leftmark{}

\maketitle

\def\leftmark{}
\section{Introduction}
Device-independent tests of assumptions depend {\em only\/} on some input-output behaviour, whereas the ``internals'' of the concrete physical systems, as well as the used devices are ignored.
Examples of such assumptions are free randomness, the impossibility of superluminal signaling, or predefined causal structures~\cite{Bell,Scarani:2002br,Colbeck:2012ie,Wood:2015jf}.
A combination of basic assumptions leads to inequalities composed of probabilities to observe certain outcomes.
Whenever a theory or an experiment {\em violates\/} a device-independent inequality, we can apply the contrapositive, and thus are forced to {\em drop\/} at least one of the assumptions made.
One of the most prominent findings of this line of reasoning is Bell-nonlocality~\cite{Bell:1964ws}, which is not only of fundamental interest but also leads to a series of device-independent protocols in cryptography~\cite{Mayers:1998,Barrett:2005ej,Colbeck:2006vo,Pironio:2010bu,Colbeck:gb,Scarani:2012vq,Vazirani:2014hu}.

\subsection{Historical background on space, time, and causality}

The debate about as how fundamental space and time
are to be seen has a long history in within natural philosophy. In pre-Socratic
time, the opposite standpoints on the question have arisen in the views of
{\em Parmenides\/} as opposed to {\em Heraclitus\/}: For the latter, the stage
set by a fundamental space-time structure is where the play of permanent
change | for him synonymous to existence | happens. Parmenides' world view,
on the other hand, is static and such that space and, in particular, time emerge
only subsequently and only subjectively. The described opposition can be seen
as a predecessor of the famous debate between {\em Newton\/} and {\em Leibniz\/}~\cite{Leibniz:2000tq},
centuries later. Whereas Newton starts from an initially given and static space-time,
Leibniz was criticizing that view: Space, for instance, is for him merely {\em relational
and not absolute}. The course of occidental science decided to go for Newton's
(overly successful) picture, until Leibniz' relational view was finally adopted by
{\em Mach}. Indeed, Mach's principle states that {\em inertial forces are purely relational},
and it was the crystallization point of {\em Einstein\/}'s general relativity although
the latter did, in the end, not satisfy the principle; however, it does propose a
{\em dynamic\/} space-time structure (still absolute, though) in which, additionally,
space and time become closely intertwined where they have been seen independently
in Newton's picture.

The absolute space-time resulting in general relativity has been challenged by the
following observations. First, {\em G\"odel\/}~\cite{Godel:1949eb} showed the possibility of solutions
to the general-relativistic field equations corresponding to closed space-time
curves. Second, {\em Bell\/}'s non-local correlations~\cite{Bell:1964ws} do not seem to be well
explainable according to {\em Reichenbach's principle\/}~\cite{Reichenbach:1956vl}, stating that any correlation
between two space-time events in a causal structure must be due to a common cause
or a direct influence from one event to the other. These facts can serve as a motivation
to drop the causal structure as being fundamental in the first place. That idea, as
sketched here, is not as novel as it may seem to be. In 1913, Russell~\cite{Russell:2006vn} wrote that
``the law of causality [\ldots] is a relic of a bygone age, surviving, like the monarchy,
only because it is erroneously supposed to do no harm.''
In the view adopted here, the causal space-time structure arises only {\em together with
instead of prior to\/} the pieces of classical information coming to existence,~{\it e.g.},
in the context of a quantum-measurement process\footnote{The key to understanding
	{\em how exactly\/} this happens is perhaps hidden in thermodynamics and a suitable
interpretation of quantum theory.}. A consequence of taking that standpoint is that the definition
by Colbeck and Renner~\cite{Colbeck:2011hw} of freeness of randomness in terms of causal structure can be
turned around: May the {\em free random bit\/} be the fundamental concept from which
the usual space-time causal structure emerges?

\subsection{Results}
This article studies {\em causal inequalities\/} that are derived from the two basic assumptions of {\em free randomness\/} and {\em predefined causal structures}, and their violations.
The inequalities are obeyed by input-output behaviours between multiple parties that are {\em consistent\/} with a predefined causal structure.
Any violation of such inequalities forces us to drop at least one of the two assumptions.
Indeed, in Section~\ref{sec:causality} we define causal order {\em based\/} on free randomness;
in that perspective at least, a fundamental causal structure {\em seems unnecessary}.
In theory, if one drops the assumption of a global time and adheres to {\em logical consistency only}, such inequalities can be violated --- quantumly (see Section~\ref{sec:quantum}) as well as classically (see Section~\ref{sec:classical}).
Thus, we are lead to the statement: {\em The assumption of a predefined causal structure is not a logical necessity}.

In Section~\ref{sec:fixedpoint} we show a relation between logical consistency and the {\em uniqueness of fixed-points\/} in functions.
By pushing this line of research further, we obtain a new class of circuits that are logically consistent, yet where the gates can be connected arbitrarily (see Section~\ref{sec:circuit}).
Such circuits describe a new model of computational that is strictly more powerful than the standard circuit model.
We conclude the work with a list of open questions.

\subsection{Related work}
Hardy~\cite{Hardy:2005wj,Hardy:2007bk} challenged the notion of a global time in quantum theory.
His main motivation is to reconcile quantum theory with general relativity.
While indeterminism, a feature of the former theory, is absent in the latter, the latter theory is more general compared to the former as its space-time is dynamic.
Thus, a theory that is probabilistic and has a dynamic space-time is a reasonable candidate for quantum gravity.
Chiribella, D'Ariano, and Perinotti~\cite{Chiribella:2009bh} and Chiribella, D'Ariano, Perinotti, and Valiron~\cite{Chiribella:2013bk} introduced the notion of ``quantum combs,'' which are higher-order transformations, {\it e.g.} transformations from operations to operations.
An interesting feature of these ``quantum combs'' is that they allow for superpositions of causal orders~\cite{Chiribella:2012jg}.
Such superpositions have lead to a computational advantage in certain tasks~\cite{Chiribella:2012jg,Colnaghi:2012dv,Araujo:2014kf,Feix:2015ww}.
In general, ``quantum combs'' can also describe resources beyond superpositions of causal orders.
Another framework to study such resources was introduced by Oreshkov, Costa, and Brukner~\cite{Oreshkov:2012uh} (see also~\cite{Brukner:2014if}).

\section{Definitions of causal relations and orders}
\label{sec:causality}
We describe causal relations between random variables, and in a next step, between parties.
Since we are interested in theories where causal structures are not fundamental, we distinguish between {\em input\/} and {\em output\/} random variables such that a structure emerges from these notions.
\begin{definition}[Input and output]
	\rm
	An {\em input\/} random variable carries a hat, {\it e.g.},~$\hat A$.
	The distribution over input random variables does not need to be specified.
	An {\em output\/} random variable is denoted by a single letter, {\it e.g.},~$X$.
	Output random variables come with a distribution that is conditioned on the inputs, {\it e.g.},~$P_{X|\hat A}$.
\end{definition}
This difference allows us to define {\em causal future\/} and {\em causal past\/} for random variables.
\begin{definition}[Causal future and causal past]
	\label{def:causality}
	\rm
	Let~$\hat A$ be an input and~$X$ be an output random variable.
	If and only if~$\hat A$ is correlated with~$X$, {\it i.e.},~$\exists P_{\hat A}:\,P_{\hat A}P_X\not=P_{\hat A}P_{X|\hat A}$, then~$\hat A$ is said to be {\em in the causal past of~$X$\/} and, equivalently~$X$ is said to be {\em in the causal future of~$\hat A$}.
	Furthermore,~$\hat A$ is called the {\em cause\/} and~$X$ is called the {\em effect}.
	This relation is denoted by~$\hat A\preceq X$.
\end{definition}
Such a definition follows the {\em interventionists'\/} approach to causality, {\it e.g.,} as defined by Woodward.
The intuition behind this definition is that we are allowed to manipulate only certain physical systems.
If such a manipulation influences another physical system, then the former manipulation {\em causes\/} the latter (see Figure~\ref{fig:causality}).
\begin{figure}
	\centering
	\begin{tikzpicture}
		\node[draw,shape=rectangle,minimum width=2cm,minimum height=1cm] (S) {};
		\node[right=0.3cm of S.west,inner sep=0pt,outer sep=0pt] (St) {$X$};
		\node[left=0.3cm of S.east,draw,shape=rectangle,minimum width=0.8cm,minimum height=0.3cm] (M) {};
		\draw (M.east) arc (70:110:1.18);
		\draw[-] (M.south)++(0cm,0.05cm) -- ++(0.3cm,0.2cm);
		\node[draw,shape=rectangle,minimum width=2cm,minimum height=1cm,right=of S,inner sep=0pt,outer sep=0pt] (T) {};
		\node[right=0.3cm of T.west,inner sep=0pt,outer sep=0pt] (Tt) {$\hat A$};
		\node[left=0.4cm of T.east,draw,shape=circle,inner sep=0pt,outer sep=0pt,minimum size=0.1cm] (K) {};
		\node[draw,shape=circle,inner sep=0pt,outer sep=0pt,minimum size=0.4cm] (K2) at (K) {};
		\foreach \angle in {
			-13+0*30,
			-13+1*30,
			-13+2*30,
			-13+3*30,
			-13+4*30,
			-13+5*30,
			-13+6*30,
			-13+7*30,
			-13+8*30,
			-13+9*30,
			-13+10*30,
			-13+11*30,
			-13+12*30
		}
		{
			\draw (K.center)++(\angle:0.05cm) -- +(\angle:0.15cm);
		}
		\draw[->] (T.west) -- (S.east);
	\end{tikzpicture}
	\caption{If the input random variable~$\hat A$ is manipulated and the output random variable~$X$ is correlated to~$\hat A$, then~$\hat A$ is the {\em cause\/} of~$X$.}
	\label{fig:causality}
\end{figure}
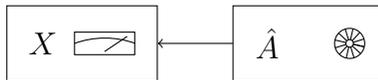
This allows us to derive the causal structure from observed correlations --- the causal structure {\em emerges\/} from the correlations.
Note that this definition does not allow an input to be an effect, and does not allow an output to be a cause.
For further studies, we define {\em parties\/} and causal relations on parties.
\begin{definition}[Party]
	\rm
	A {\em party\/}~$S=(\hat A,X,\mathcal{E})$ with~$\mathcal{E}:\hat A\times I_S\rightarrow X\times O_S$,
	consists of an input random variable~$\hat A$, an output random variable~$X$, and a map~$\mathcal{E}$ that maps~$\hat A$ together with a physical system that~$S$ receives from the environment to~$X$ and a physical system that is returned to the environment (see Figure~\ref{fig:party}).
	\begin{figure}
		\centering
		\begin{tikzpicture}
			\node[draw,rectangle,minimum width=1cm,minimum height=1cm] (S) {$\mathcal{E}$};
			\draw[->] (S.150) -- ++(-0.2,0) node [left] {$X$};
			\draw[<-] (S.210) -- ++(-0.2,0) node [left] {$\hat A$};
			\draw[->] (S.90) -- ++(0,0.5) node [above] {$O_S$};
			\draw[<-] (S.270) -- ++(0,-0.5) node [below] {$I_S$};
		\end{tikzpicture}
		\caption{Party~$S=(\hat A,X,\mathcal{E})$ with~$\mathcal{E}:\hat A\times I_S\rightarrow X\times O_S$, where~$I_S$ is a physical system obtained from the environment and~$O_S$ is a physical system returned to the environment.}
		\label{fig:party}
	\end{figure}
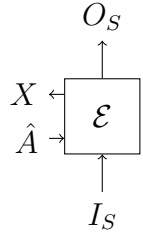
	A party interacts at most once with the environment, where we consider one reception and one transmission of a system as a single interaction.
\end{definition}
Now, we can define causal relations between parties.
\begin{definition}[Causal future and causal past for parties]
	\rm
	Let~$R=(\hat A,X,\mathcal{E})$ and let~\mbox{$S=(\hat B,Y,\mathcal{F})$} be two parties.
	We say~$R$ is in the {\em causal past\/} of~$S$ or~$S$ is in the {\em causal future\/} of~$R$ if and only if~$\hat A$ is correlated with~$Y$ and~$\hat B$ is uncorrelated with~$X$.
	This relation is denoted by~$R\preceq S$.
\end{definition}

\subsection{Predefined vs.~indefinite causal order}
In this work we distinguish between {\em predefined\/} and {\em indefinite\/} causal orders.
In particular, we will show that if one defines causality {\em based on the inputs and outputs\/} (see Definition~\ref{def:causality}), then correlations that are not compatible with a predefined causal order could arise.
\begin{definition}[Compatibility with two-party predefined causal order]
	\label{def:2partypredefined}
	\rm
	Let~$R=(\hat A,X,\mathcal{E})$ and~$S=(\hat B,Y,\mathcal{F})$ be two parties.
	A conditional probability distribution~$P_{X,Y|\hat A,\hat B}$ is called {\em consistent with two-party predefined causal order\/} if and only if the conditional probability distribution can be written as a convex combination of the orderings~\mbox{$R\preceq S$} and~\mbox{$S\preceq R$}, {\it i.e.},
	\begin{eqnarray}
		P_{X,Y|\hat A,\hat B}=pP_{X|\hat A}P_{Y|\hat A,\hat B,X} + (1-p)P_{X|\hat A,\hat B,Y}P_{Y|\hat B}
		\nonumber
		\,,
	\end{eqnarray}
	for some~$0\leq p\leq 1$.
	The first term in the sum represents~$R\preceq S$ and the second term represents~$S\preceq R$.
\end{definition}
This definition determines a polytope of probability distributions~$P_{X,Y|\hat A,\hat B}$ that are compatible with two-party predefined causal order.
All facets for binary random variables were recently enumerated~\cite{araujo14b}.

For three parties or more, a distribution that is compatible with a predefined causal order can be more general than a distribution that can be written as a convex combination of all orderings.
The reason for this is that a party in the causal past of some other parties could in principle influence everything which lies in its causal future --- therefore, it could also influence the {\em causal order\/} of the parties in its causal future.
\begin{definition}[Compatibility with three-party predefined causal order]
	\label{def:predefined}
	\rm
	Consider the three parties~$R=(\hat A,X,\mathcal{E})$,~$S=(\hat B,Y,\mathcal{F})$, and~$T=(\hat C,Z,\mathcal{G})$.
	A conditional probability distribution~$P_{X,Y,Z|\hat A,\hat B,\hat C}$ is called {\em consistent with predefined causal order\/} if and only if the probability distribution can be written as a convex combination of all orderings where one party~$Q$ is in the causal past of the other two, and where the causal order of these two parties is determined by~$Q$.
\end{definition}
A generalized version of Definition~\ref{def:predefined} to any number of parties can be found in~\cite{Oreshkov:2015vs}.
\begin{lemma}
	{\bf\cite{Baumeler:2015wx}}
	{\rm (Necessary condition for predefined causal order)}
	\label{lemma:necessarycondition}
	A {\em necessary\/} condition for predefined causal order is that {\em at least one party is not in the causal future of any party}.
\end{lemma}
If a conditional distribution is {\em incompatible\/} with predefined causal order, then we call it {\em indefinite}.

\begin{example}[One-way signaling]
	\rm
	Let~$R=(\hat A,X,\mathcal{E})$ and~$S=(\hat B,Y,\mathcal{F})$ be two parties.
	The probability distribution over binary random variables
	\begin{equation*}
		P_{X,Y|\hat A,\hat B}(x,y,a,b)=\cases{
			1/2&for $x=b\wedge y=0$\\
			1/2&for $x=b\wedge y=1$\\
			0&otherwise
		}
	\end{equation*}
	is {\em compatible\/} with predefined causal order, because it can be written as
	\begin{eqnarray*}
		P_{X,Y|\hat A,\hat B}&=P_{X|\hat A,\hat B,Y}P_{Y|\hat B}
		\,,
	\end{eqnarray*}
	with
	\begin{equation*}
		P_{X|\hat A,\hat B,Y}(x,a,b,y)=\cases{
			1&for~$x=b$\\
			0&otherwise,}
	\end{equation*}
	and
	\begin{equation*}
		P_{Y|\hat B}(y,b)=1/2
		\,.
	\end{equation*}
\end{example}

\begin{example}[Two-way signaling]
	\rm
	Let~$R=(\hat A,X,\mathcal{E})$ and~$S=(\hat B,Y,\mathcal{F})$ be two parties.
	The probability distribution over binary random variables
	\begin{equation*}
		P_{X,Y|\hat A,\hat B}(x,y,a,b)=\cases{
			1&for $x=b\wedge y=a$\\
			0&otherwise
		}
	\end{equation*}
	is {\em incompatible\/} with predefined causal order (has an {\em indefinite\/} causal order), because it {\em cannot\/} be written as described in Definition~\ref{def:2partypredefined}.
\end{example}

\section{Assuming quantum theory locally}
\label{sec:quantum}
In 2012, Oreshkov, Costa, and Brukner~\cite{Oreshkov:2012uh} introduced the process-matrix framework for quantum correlations without predefined causal order.
In that framework, a party~$R=(\hat A,X,\mathcal{E}:\hat A\times I_R\rightarrow X\times O_R)$ receives a quantum state on the Hilbert space~$I_R$ from the environment and returns a quantum state on the Hilbert space~$O_R$ to the environment.
The map~$\mathcal{E}$ is completely positive, because we assume the validity of quantum theory within the laboratories of every party.
Let\footnote{Throughout this work we use bold letters for vectors and matrices.}~$\bm{\mathcal{E}}_{x,a}$ be the corresponding Choi-Jamio{\l}kowski~\cite{Jamioikowski:1972ke,Choi:1975el} map which is an element in~$I_R\otimes O_R$, where~$a$ is the input value and~$x$ is the output value.
Let~$S=(\hat B,Y,\mathcal{F}:\hat B\times I_S\rightarrow Y\times O_S)$ be another party with Choi-Jamio{\l}kowski map~$\bm{\mathcal{F}}_{y,b}$.
The most general probability distribution~$P_{X,Y|\hat A,\hat B}$ that is linear in the local operations is
\begin{eqnarray*}
	P_{X,Y|\hat A,\hat B}(x,y,a,b)=\mathrm{Tr}\left((\bm{\mathcal{E}}_{x,a}\otimes \bm{\mathcal{F}}_{y,b})W\right)
	\,,
\end{eqnarray*}
where~$W$ is a matrix living in the Hilbert space~$I_R\otimes O_R\otimes I_S\otimes O_S$.
The matrix~$W$ is called {\em process matrix}.

Since~$P_{X,Y|\hat A,\hat B}$ is designed to be a probability distribution, and since both parties can arbitrarily choose their local operations~$\mathcal{E}$ and~$\mathcal{F}$, the following two conditions must be satisfied:
\begin{eqnarray}
	\forall \mathcal{E},\mathcal{F},x,y,a,b:&\,P_{X,Y|\hat A,\hat B}(x,y,a,b)\geq 0\,,\label{eq:nonneg}\\
	\forall \mathcal{E},\mathcal{F},a,b:&\,\sum_{x,y}P_{X,Y|\hat A,\hat B}(x,y,a,b)=1\,.\label{eq:totprob}
\end{eqnarray}
\begin{definition}[Logically consistent process matrix]
	\rm
	We call a process matrix~$W$ {\em logically consistent\/} if and only if~$W$ satisfies the conditions~(\ref{eq:nonneg}) and~(\ref{eq:totprob}).
\end{definition}

Condition~(\ref{eq:nonneg}) implies that~$W$ must be a completely positive trace-preserving map from~$O_R\otimes O_S$ to~$I_R\otimes I_S$.
Therefore, we can interpret~$W$ as a quantum channel (see Figure~\ref{fig:W}).
\begin{figure}
	\centering
        \begin{tikzpicture}
		\node[draw,rectangle,minimum height=1cm,minimum width=1cm] (R) {$\mathcal{E}$};
                \draw[->] (R.150) -- ++(-0.2,0) node [left] {$X$};
                \draw[<-] (R.210) -- ++(-0.2,0) node [left] {$\hat A$};
                \node[draw,rectangle,minimum width=1cm,minimum height=1cm,right=2cm of R] (W) {$W$};
		\node[draw,rectangle,minimum height=1cm,minimum width=1cm,right=2cm of W] (S) {$\mathcal{F}$};
                \draw[->] (S.30) -- ++(0.2,0) node [right] {$Y$};
                \draw[<-] (S.330) -- ++(0.2,0) node [right] {$\hat B$};
                \draw[->] (R) to[out=80,in=100] (W);
                \draw[<-] (R) to[out=280,in=260] (W);
                \draw[->] (S) to[out=100,in=80] (W);
                \draw[<-] (S) to[out=260,in=280] (W);
                \node (SO) at ([shift={(-0.3cm,0.8cm)}]R) {$O_R$};
                \node (SI) at ([shift={(-0.3cm,-0.8cm)}]R) {$I_R$};
                \node (TO) at ([shift={(0.3cm,0.8cm)}]S) {$O_S$};
                \node (TI) at ([shift={(0.3cm,-0.8cm)}]S) {$I_S$};
        \end{tikzpicture}
	\caption{Two-party process matrix~$W$ as a channel from the Hilbert space~$O_R\otimes O_S$ to~$I_R\otimes I_S$.}
	\label{fig:W}
\end{figure}
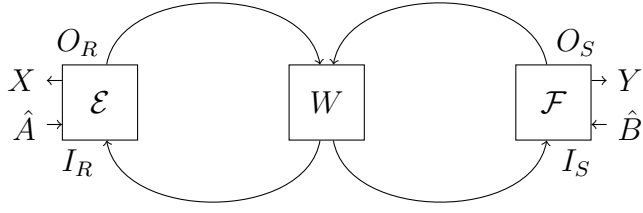
If we assumed a global causal structure (what we do not), then~$W$ would be a {\em back-in-time\/} channel.
Rather, it can be interpreted as the {\em environment which lies outside space-time\/} --- the causal structure is designed by~$W$ and emerges from the correlations, as will become clear later.
First, we can understand~$W$ as a generalized notion of a state and a channel: Whereas it describes a {\em quantum state\/} in Example~\ref{ex:state}, it models a {\em quantum channel\/} in Example~\ref{ex:channel}.
\begin{example}[Representation of a quantum state]
	\rm
	\label{ex:state}
	This logically consistent process matrix
	\begin{eqnarray*}
		W_\mathrm{state}=\rho_{I_S,I_R}\otimes\mathds{1}_{O_S,O_R}
	\end{eqnarray*}
	describes the {\em quantum state\/}~$\rho$ that is sent to both parties.
\end{example}
\begin{example}[Representation of a quantum channel]
	\rm
	\label{ex:channel}
	This logically consistent process matrix
	\begin{eqnarray*}
		W_\mathrm{channel}=\mathds{1}_{I_R}\otimes|\Psi\rangle\langle\Psi|_{O_R,I_S}\otimes\mathds{1}_{O_S}
		\,,
	\end{eqnarray*}
	with~$|\Psi\rangle=(|0,0\rangle+|1,1\rangle)/\sqrt{2}$, describes a {\em qubit channel\/} from party~$R$ to party~$S$.
\end{example}

Besides quantum states and quantum channels, a logically consistent process matrix can also describe {\em superpositions of quantum channels}.
For instance, we could have a channel from~$R$ to~$S$ to~$T$ superposed with a channel from~$S$ to~$R$ to~$T$, where~$T=(\hat C,Z,\mathcal{G}:\hat C\times I_T\times I_{T'}\rightarrow Z)$ is a party which {\em does not\/} return a system to the environment, and receives two systems from the environment (the target on~$I_T$ and the control on~$I_{T'}$) (see Example~\ref{ex:superposition}).
\begin{example}[Superposition of channels]
	\rm
	\label{ex:superposition}
	This logically consistent process matrix
	\begin{eqnarray*}
		W_\mathrm{superposedchannel}=|w\rangle\langle w|\,,
	\end{eqnarray*}
	with
	\begin{eqnarray*}
		|w\rangle = \frac{
		|0,\Psi,\Psi,0\rangle_{I_R,O_R,I_S,O_S,I_T,I_{T'}}
		+
		|0,\Psi,\Psi,1\rangle_{I_S,O_S,I_R,O_R,I_T,I_{T'}}
	}{\sqrt 2}
	\end{eqnarray*}
	describes a superposition of a channel from~$R$ to~$S$ to~$T$ and a channel from~$S$ to~$R$ to~$T$.
\end{example}
The process matrix from Example~\ref{ex:superposition} can be used to solve certain tasks more efficiently.
Suppose you are given two black boxes~$\mathrm{B}$ and~$\mathrm{C}$ that act on qubits, and you are guaranteed that~$\mathrm{B}$ and~$\mathrm{C}$ either {\em commute\/} or {\em anti-commute}.
In the standard circuit model, you would need to query each box twice in order to determine the (anti)commutativity.
Whereas by using the process matrix~$W_\mathrm{superposedchannel}$, a single query suffices~\cite{Chiribella:2012jg,Colnaghi:2012dv,Araujo:2014kf,Araujo:2015ky,Procopio:2015iw}.

\subsection{Non-causal process matrices}
Whilst the above examples of process matrices are {\em compatible with predefined causal order}, some logically consistent process matrices lead to correlations that cannot be obtained in a world with a predefined causal ordering of the parties --- such process matrices are called {\em non-causal}.
\begin{definition}[Causal and non-causal process matrices]
	\rm
	A process matrix~$W$ is called {\em causal\/} if and only if for any choice of operations of the parties the resulting probability distribution {\em is compatible with predefined causal order\/} (according to Definition~\ref{def:predefined}).
	Otherwise, it is called {\em non-causal}.
\end{definition}

To show that a process matrix is non-causal, we define a game, give an upper bound on the winning probability for this game under the assumption of a predefined causal order, and show that this bound can be violated in the process-matrix framework.
\begin{game}
	{\bf\cite{Oreshkov:2012uh}}
	{\rm (Two-party non-causal game)}
	\rm
	\label{game:2party}
	Let~$R=(\hat A,X,\mathcal{E})$ and~$S=((\hat B,\hat B'),Y,\mathcal{F})$ be two parties that aim at maximizing
	\begin{eqnarray*}
		p_\mathrm{succ}=\frac{1}{2}\left(\Pr\left(X=\hat B\,|\,b'=0\right)+\Pr\left(Y=\hat A\,|\,b'=1\right)\right)
		\,,
	\end{eqnarray*}
	where all random variables are binary and where all inputs are uniformly distributed.
	Informally, if~$b'=0$, then party~$S$ is asked to send her input to~$R$, otherwise, party~$R$ is asked to send her input to~$S$.
\end{game}
We give an upper bound on the success probability of Game~\ref{game:2party} under the assumption of a predefined causal order.
\begin{theorem}
	{\bf\cite{Oreshkov:2012uh}}
	{\rm (Upper bound on success probability of Game~\ref{game:2party})}
	\label{thm:2party}
	Under the assumption of a predefined causal order, Game~\ref{game:2party} can at best be won with probability~$3/4$.
\end{theorem}
\begin{proof}
	By Lemma~\ref{lemma:necessarycondition}, at least one party is not in the causal future of any party.
	Without loss of generality, let~$R$ be this party.
	Then,~$R$ can only make a random guess, which means~$\Pr(X=\hat B\,|\,b'=0)=1/2$.
	Therefore, we obtain the upper bound
	\begin{eqnarray*}
		p_\mathrm{succ}=\frac{1}{2}\left(\frac{1}{2}+\Pr\left(Y=\hat A\,|\,b'=1\right)\right)\leq\frac{3}{4}
		\,.
	\end{eqnarray*}
\end{proof}
Then again, this upper bound given by Theorem~\ref{thm:2party} can be violated in the process-matrix framework.
\begin{theorem}
	{\bf \cite{Oreshkov:2012uh}}
	\label{thm:2partyviolation}
	Game~\ref{game:2party} can be won with probability~$(2+\sqrt 2)/4$ if we drop the assumption of a global causal order.
\end{theorem}
The logically consistent process matrix
\begin{eqnarray*}
	W=\frac{1}{4}\left( \mathds{1}+\frac{\mathds{1}_{I_R}(\sigma_z)_{O_R}(\sigma_z)_{I_S}\mathds{1}_{O_S}
	+(\sigma_z)_{I_R}\mathds{1}_{O_R}(\sigma_x)_{I_S}(\sigma_z)_{O_S} }{\sqrt 2} \right)
	\,,
\end{eqnarray*}
where we omit the~$\otimes$ symbols for better presentation,
and the local operations
\begin{eqnarray*}
	\bm{\mathcal{E}}_{x,a}& = &\frac{1}{4}\left( \mathds{1}+(-1)^x\sigma_z \right)_{I_R}\otimes\left( \mathds{1}+(-1)^a\sigma_z \right)_{O_R}\,,\\
	\bm{\mathcal{F}}_{y,b,b'=0}& = &\frac{1}{4}\left( \mathds{1}+(-1)^y\sigma_x \right)_{I_S}\otimes\left( \mathds{1}+(-1)^{b+y}\sigma_z \right)_{O_S}\,,\\
	\bm{\mathcal{F}}_{y,b,b'=1}& = &\frac{1}{2}\left( \mathds{1}+(-1)^z\sigma_x \right)_{I_S}\otimes\rho_{O_S}\,,
\end{eqnarray*}
for an arbitrary~$\rho$, can be used to violate the bound given by Theorem~\ref{thm:2party} up to~$(2+\sqrt 2)/4$~\cite{Oreshkov:2012uh}.
Brukner~\cite{Brukner:2014vo} proved under certain assumptions that within the process-matrix framework, this violation cannot be exceeded.

This raised the question whether {\em non-causal\/} logically consistent process matrices in the classical realm also exist --- which was answered negatively in the two-party case~\cite{Oreshkov:2012uh,Costa:2013vc}.  
However, for three parties or more, this is not true anymore; this is shown in Section~\ref{sec:classical}.

Logically consistent process matrices can be understood as a new resource for quantum operations.
Oddly enough, these resources cannot be composed: If we take two logically consistent process matrices~$W_1$ and~$W_2$, where the former is a quantum channel from~$R$ to~$S$ and the latter is a quantum channel from~$S$ to~$R$, then~$W_1\otimes W_2$ is {\em not\/} logically consistent;~$R$ could in principle alter the state on~$I_R$ she receives from the environment, leading to a causal paradox.
The impossibility of composing process matrices within the model reflects the fact that a process matrix is supposed to describe the environment as a {\em whole}.

\section{Assuming classical probability theory locally}
\label{sec:classical}
The {\em classical\/} analogue of the process-matrix framework was recently developed~\cite{Baumeler:2015wx}, and has been found to give rise to non-causal correlations for three parties or more.
It is the classical analogue in the sense that, instead of assuming the validity of quantum theory locally, classical probability theory is assumed to hold locally.
Let~\mbox{$R=(\hat A,X,\mathcal{E}:\hat A\times I_R\rightarrow X\times O_R)$} be a party.
Here, the spaces~$I_R$ and~$O_R$ are not Hilbert spaces (as in the process-matrix framework), but describe random variables.
Therefore, the local operation~$\mathcal{E}$ of a party is a conditional probability distribution~$P_{X,O_R|\hat A,I_R}$; locally we assume the validity of probability theory as opposed to quantum theory.
Let~$S=(\hat B,Y,\mathcal{F})$ with~$\mathcal{F}:\hat B\times I_S\rightarrow Y\times O_S$ and~$T=(\hat C,Z,\mathcal{G})$ with~$\mathcal{G}:\hat C\times I_T\rightarrow Z\times O_T$ be two further parties.
The most general probability distribution~$P_{X,Y,Z,I_R,I_S,I_T,O_R,O_S,O_T|\hat A,\hat B,\hat C}$ that is linear in the local operations is
\begin{eqnarray}
	\mathcal{E}(x,o_R,a,i_R)
	\mathcal{F}(y,o_S,b,i_S)
	\mathcal{G}(z,o_T,c,i_T)
	E(i_R,i_S,i_T,o_R,o_S,o_T)
	\label{eq:classical}
	\,,
\end{eqnarray}
where~$E$ is called {\em classical process}.
Since all three parties can arbitrarily choose their local operations~$\mathcal{E}$,~$\mathcal{F}$, and~$\mathcal{G}$, the following conditions must be satisfied (we use~$\bm i$ as shorthand expression for~$(i_R,i_S,i_T)$, likewise for~$\bm o$,~$\bm I$, and~$\bm O$):
\begin{eqnarray}
	\forall\mathcal{E},\mathcal{F},\mathcal{G},x,y,z,\bm i,\bm o,a,b,c:\,P_{X,Y,Z,\bm I,\bm O|\hat A,\hat B\hat C}(x,y,z,\bm i,\bm o,a,b,c)\geq 0\,,\label{eq:clnonneg}\\
	\forall\mathcal{E},\mathcal{F},\mathcal{G},a,b,c:\,\sum_{x,y,z,\bm i,\bm o}P_{X,Y,Z,\bm I,\bm O|\hat A,\hat B\hat C}(x,y,z,\bm i,\bm o,a,b,c)=1\,.\label{eq:cltotprob}
\end{eqnarray}
\begin{definition}[Logically consistent clasical process]
	\rm
	We call a classical process~$E$ {\em logically consistent\/} if and only if~$E$ satisfies the conditions~(\ref{eq:clnonneg}) and~(\ref{eq:cltotprob}).
\end{definition}
Condition~(\ref{eq:clnonneg}) implies that the classical process~$E$ is a conditional probability distribution~$P_{I_R,I_S,I_T|O_R,O_S,O_T}$.
Therefore,~$E$ can be interpreted in the same way as process matrices: it maps the systems~$O_R$,~$O_S$, and~$O_T$ to~$I_R$,~$I_S$, and~$I_T$ (see Figure~\ref{fig:E}).
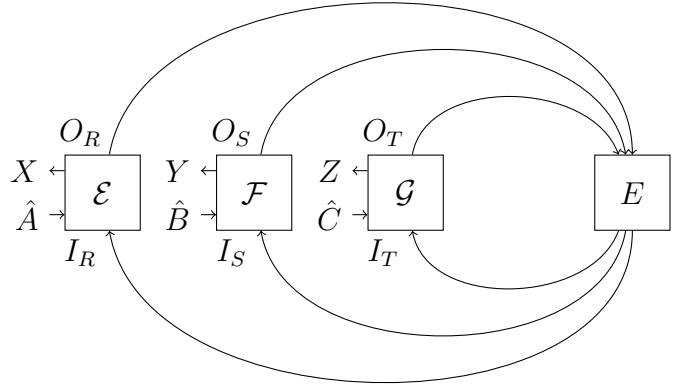
\begin{figure}
	\centering
        \begin{tikzpicture}
		\node[draw,rectangle,minimum height=1cm,minimum width=1cm] (R) {$\mathcal{E}$};
                \draw[->] (R.150) -- ++(-0.2,0) node [left] {$X$};
                \draw[<-] (R.210) -- ++(-0.2,0) node [left] {$\hat A$};
		\node[draw,rectangle,minimum height=1cm,minimum width=1cm,right=1cm of R] (S) {$\mathcal{F}$};
                \draw[->] (S.150) -- ++(-0.2,0) node [left] {$Y$};
                \draw[<-] (S.210) -- ++(-0.2,0) node [left] {$\hat B$};
		\node[draw,rectangle,minimum height=1cm,minimum width=1cm,right=1cm of S] (T) {$\mathcal{G}$};
                \draw[->] (T.150) -- ++(-0.2,0) node [left] {$Z$};
                \draw[<-] (T.210) -- ++(-0.2,0) node [left] {$\hat C$};
                \node[draw,rectangle,minimum width=1cm,minimum height=1cm,right=2cm of T] (E) {$E$};
                \draw[->] (R) to[out=80,in=90] (E);
                \draw[<-] (R) to[out=280,in=270] (E);
                \draw[->] (S) to[out=80,in=100] (E);
                \draw[<-] (S) to[out=280,in=260] (E);
                \draw[->] (T) to[out=80,in=110] (E);
                \draw[<-] (T) to[out=280,in=250] (E);
                \node (RO) at ([shift={(-0.3cm,0.8cm)}]R) {$O_R$};
                \node (RI) at ([shift={(-0.3cm,-0.8cm)}]R) {$I_R$};
                \node (SO) at ([shift={(-0.3cm,0.8cm)}]S) {$O_S$};
                \node (SI) at ([shift={(-0.3cm,-0.8cm)}]S) {$I_S$};
                \node (TO) at ([shift={(-0.3cm,0.8cm)}]T) {$O_T$};
                \node (TI) at ([shift={(-0.3cm,-0.8cm)}]T) {$I_T$};
        \end{tikzpicture}
	\caption{Three-party classical process~$E$ as a channel from~$(O_R,O_S,O_T)$ to~$(I_R,I_S,I_T)$.}
	\label{fig:E}
\end{figure}

We rewrite the conditions~(\ref{eq:clnonneg}) and~(\ref{eq:cltotprob}) by using stochastic matrices.
This helps to check whether a classical process is logically consistent or not.
We write~$\bm{E}$ to denote the stochastic matrix that models the classical process~$E$.
Then, the non-negativity condition~(\ref{eq:clnonneg}) becomes
\begin{eqnarray}
	\forall m,n:\,\bm{E}_{m,n}\geq 0
	\label{eq:mnonneg}
	\,,
\end{eqnarray}
where~$\bm{E}_{m,n}$ are the matrix elements.
For simplicity, we fix all inputs~$(\hat A,\hat B,\hat C)$ to~$(a,b,c)$ and consider the operations~\mbox{$\mathcal{E'}:I_R\rightarrow O_R$} only, and likewise for the other parties.
Let~$\bm{\mathcal{E'}}$ be the corresponding stochastic matrix.
The probability that the operation~$\mathcal{E'}$ produces~$o_R$ conditioned on~$i_R$~is
\begin{eqnarray*}
	P_{O_R|I_R}(o_R,i_R)=\bm{o_R}^T \bm{\mathcal{E'}}\bm{i_R}
	\,,
\end{eqnarray*}
where~$\bm{o_R}$ is the stochastic vector\label{p:sv} that models~$o_R$, {\em e.g.}, for a binary random variable~$O_R$ the value~$o_R=0$ is modeled by
\begin{eqnarray*}
	\bm{o_R}=
	\left(
	\begin{array}{c}
		1\\0
	\end{array}
	\right)
	\,.
\end{eqnarray*}
Having this, we rewrite condition~(\ref{eq:cltotprob}) as
\begin{eqnarray*}
	\forall \mathcal{E},\mathcal{F},\mathcal{G}:\,&\sum_{\bm i,\bm o}&
	\left(\bm{i_R}\otimes \bm{i_S}\otimes \bm{i_T}\right)^T
	\bm{E}
	\left(\bm{o_R}\otimes \bm{o_S}\otimes \bm{o_T}\right)\\
	&\times&\left(\bm{o_R}^T \bm{\mathcal{E'}} \bm{i_R}\right)
	\left(\bm{o_S}^T \bm{\mathcal{F'}} \bm{i_S}\right)
	\left(\bm{o_T}^T \bm{\mathcal{G'}} \bm{i_T}\right)
	\\
	&=&\mathrm{Tr}\left( \bm{E}\left(\bm{\mathcal{E'}}\otimes \bm{\mathcal{F'}}\otimes \bm{\mathcal{G'}}\right) \right)\\
	&=&1
	\,.
\end{eqnarray*}
Indeed, since any local operation~$\mathcal{E'}$ can be written as a convex combination of deterministic operations, the total-probability condition
\begin{eqnarray}
	\forall \mathcal{E'},\mathcal{F'},\mathcal{G'}\in\mathcal{D}:\,\mathrm{Tr}\left( \bm{E}\left(\bm{\mathcal{E'}}\otimes \bm{\mathcal{F'}}\otimes \bm{\mathcal{G'}}\right) \right)=1
	\label{eq:trcond}
	\,,
\end{eqnarray}
where~$\mathcal{D}$ is the set of all {\em deterministic\/} local operations, is sufficient.

\subsection{Non-causal classical processes}
In analogy to the process-matrix framework, we can define {\em non-causal\/} classical processes.
\begin{definition}[Non-causal classical processes]
	\rm
	A classical process~$E$ is called {\em non-causal\/} if and only if there exists a choice of operations of the parties such that the resulting probability distribution does {\em not\/} satisfy the necessary condition for predefined causal order (see Lemma~\ref{lemma:necessarycondition}).
\end{definition}

It is known that for two parties, non-causal logically consistent classical processes do not exist~\mbox{\cite{Oreshkov:2012uh,Costa:2013vc}}.
For three parties or more, however, such processes {\em do\/} exist~\cite{Baumeler:2014cw,Baumeler:2015wx}.
We describe two non-causal logically consistent classical processes.

\begin{game}
	{\bf\cite{Baumeler:2014cw}}
	{\rm (Three-party non-causal)}
	\rm
	\label{game:3party}
	Let~$R=((\hat A,\hat M),X,\mathcal{E})$,~$S=((\hat B,\hat M),Y,\mathcal{F})$, and~$T=((\hat C,\hat M),Z,\mathcal{G})$ be three parties that aim at maximizing
	\begin{eqnarray*}
		p_\mathrm{succ2}&=&
		\frac{1}{3}\left(
			\Pr\left(X=\hat B\oplus \hat C\,|\,m=0\right)\right.\\
			&+&\left.\Pr\left(Y=\hat A\oplus \hat C\,|\,m=1\right)+
			\Pr\left(Z=\hat A\oplus \hat B\,|\,m=2\right)\right)
			\,,
	\end{eqnarray*}
	where all random variables but~$\hat M$ are binary, and where~$\hat M$ is a {\em shared\/} ternary random variable.
	All input random variables are uniformly distributed.
	Informally, if the shared random variable~$\hat M$ takes value~$0$, then party~$R$ has to guess the parity of the inputs of~$S$ and~$T$, and likewise for the alternative values the random variable~$\hat M$ can take.
\end{game}
We give an upper bound on the success probability of Game~\ref{game:3party} under the assumption of a predefined causal order.
\begin{theorem}
	{\bf\cite{Baumeler:2014cw}}
	\label{thm:3party}
	In a predefined causal order, Game~\ref{game:3party} can at best be won with probability~$5/6$.
\end{theorem}
\begin{proof}
	A necessary condition for predefined causal order is that at least one party is not in the causal future of any party (see Lemma~\ref{lemma:necessarycondition}).
	Without loss of generality, let~$R$ be this party.
	Whenever the shared trit~$\hat M$ takes the value~$0$, party~$R$ can give a random guess only.
	Therefore,~$\Pr(X=\hat B\oplus \hat C\,|\,m=0)=1/2$.
	This gives the upper bound
	\begin{eqnarray*}
		p_\mathrm{succ2}&=\frac{1}{3}\left( \frac{1}{2}+\Pr\left(Y=\hat A\oplus \hat C\,|\,m=1\right)+\Pr\left(Z=\hat A\oplus \hat B\,|\,m=2\right) \right)\\
		&\leq \frac{5}{6}
		\,.
	\end{eqnarray*}
\end{proof}
\begin{theorem}
	{\bf\cite{Baumeler:2014cw}}
	\label{thm:3partyperfect}
	Game~\ref{game:3party} can be won perfectly if we drop the assumption of a global causal order.
\end{theorem}
Consider the logically consistent classical process~$E$
\begin{eqnarray*}
	P_{\bm I|\bm O}(\bm i,\bm o)=\cases{
		1/2&for $i_R=o_T\wedge i_S=o_R\wedge i_T=o_S$\\
		1/2&for $i_R=o_T\oplus 1\wedge i_S=o_R\oplus 1\wedge i_T=o_S\oplus 1$\\
		0&otherwise.
	}
\end{eqnarray*}
This process is the uniform mixture of the identity channel from~$R$ to~$S$ to~$T$ to~$R$ and the bit-flip channel from~$R$ to~$S$ to~$T$ to~$R$ (see Figure~\ref{fig:stopro}).
\begin{figure}
\centering
\begin{tikzpicture}
	\def\r{1}
	\def\d{20}
	\def\dx{2.1}
	\def\hd{.5}

	\draw[->] (-\dx,0)++(90-\d:\r) arc (90-\d:-30+\d:\r);
	\draw[->] (-\dx,0)++(-30-\d:\r) arc (-30-\d:-150+\d:\r);
	\draw[->] (-\dx,0)++(-150-\d:\r) arc (-150-\d:-270+\d:\r);
	\draw (-\dx,0)++(90:\r) node (A) {$R$};
	\draw (-\dx,0)++(-30:\r) node (B) {$S$};
	\draw (-\dx,0)++(-150:\r) node (C) {$T$};

	\draw[->] (\dx,0)++(90-\d:\r) arc (90-\d:-30+\d:\r);
	\draw[->] (\dx,0)++(-30-\d:\r) arc (-30-\d:-150+\d:\r);
	\draw[->] (\dx,0)++(-150-\d:\r) arc (-150-\d:-270+\d:\r);
	\draw (\dx,0)++(90:\r) node (A2) {$R$};
	\draw (\dx,0)++(-30:\r) node (B2) {$S$};
	\draw (\dx,0)++(-150:\r) node (C2) {$T$};
	\draw (\dx,0)++(90-60:\r+.3) node (n1) {$\oplus 1$};
	\draw (\dx,0)++(-30-60:\r+.3) node (n2) {$\oplus 1$};
	\draw (\dx,0)++(-150-60:\r+.3) node (n3) {$\oplus 1$};

	\draw (-\hd/2,0) node (PLUS) {\LARGE $+$};
	\draw (-\dx-\r-\hd,0) node (half1) {\LARGE $\frac{1}{2}$};
	\draw (\dx-\r-\hd,0) node (half2) {\LARGE $\frac{1}{2}$};
\end{tikzpicture}
\caption{Logically consistent classical process used to perfectly win Game~\ref{game:3party}.}
\label{fig:stopro}
\end{figure}
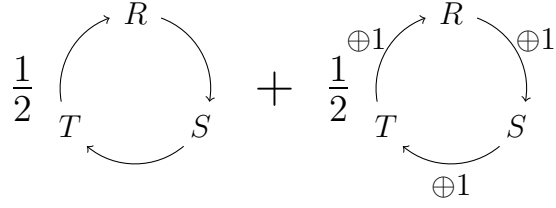
By combining it with the local operations
\begin{eqnarray*}
	\mathcal{E}(x,o_R,a,m,i_R)=\cases{
		1&for $m=0\wedge x=i_R\wedge o_R=0$\\
		1&for $m=1\wedge x=0\wedge o_R=i_R\oplus a$\\
		1&for $m=2\wedge x=0\wedge o_R=a$\\
		0&otherwise,
	}\\
	\mathcal{F}(y,o_S,b,m,i_S)=\cases{
		1&for $m=0\wedge y=0\wedge o_S=b$\\
		1&for $m=1\wedge y=i_S\wedge o_S=0$\\
		1&for $m=2\wedge y=0\wedge o_S=i_S\oplus b$\\
		0&otherwise,
	}\\
	\mathcal{G}(z,o_T,c,m,i_T)=\cases{
		1&for $m=0\wedge z=0\wedge o_T=i_T\oplus c$\\
		1&for $m=1\wedge z=0\wedge o_T=c$\\
		1&for $m=2\wedge z=i_T\wedge o_T=0$\\
		0&otherwise,
	}
\end{eqnarray*}
Game~\ref{game:3party} is won perfectly.
The operations can be understood as follows.
The party who has to make the guess simply uses the system that she receives from the environment as guess.
If the party who has to produce a guess is the next in the chain ($R$ to~$S$ to~$T$ to~$R$) from the point of view of the actual party, then the actual party returns the parity of the system she obtained from the environment and her input.
In the third case, if the party who has to guess is two steps ahead in the chain, then the actual party returns her input to the environment.

\begin{game}
	{\bf\cite{Baumeler:2015wx}}
	{\rm (Three-party non-causal game)}
	\rm
	\label{game:3partydet}
	Let~$R=(\hat A,X,\mathcal{E})$,~$S=(\hat B,Y,\mathcal{F})$, and~$T=(\hat C,Z,\mathcal{G})$ be three parties that aim at maximizing
	\begin{eqnarray*}
		p_\mathrm{succ3}&=&
		\frac{1}{2}\left(
			\Pr\left(X=\hat C,Y=\hat A,Z=\hat B\,|\,\mathrm{maj}(\hat A,\hat B,\hat C)=0\right)\right.\\
			&+&\left.\Pr\left(X=\hat B\oplus 1,Y=\hat C\oplus 1,Z=\hat A\oplus 1\,|\,\mathrm{maj}(\hat A,\hat B,\hat C)=1\right)\right)
			\,,
	\end{eqnarray*}
	where all random variables are binary, and where~$\mathrm{maj}(a,b,c)$ is the majority of~$(a,b,c)$.
	All inputs are uniformly distributed.
	In words, if the majority of the inputs is~$0$, then the parties have to guess their neighbour's input in one direction.
	Otherwise, if the majority of the inputs is~$1$, then the parties have to guess their neighbour's inverted input in the opposite direction.
\end{game}
We give an upper bound on the success probability of Game~\ref{game:3partydet} under the assumption of a predefined causal order.
\begin{theorem}
	{\bf\cite{Baumeler:2015wx}}
	\label{thm:3partydet}
	Under the assumption of a predefined causal order, Game~\ref{game:3partydet} can at best be won with probability~$3/4$.
\end{theorem}
\begin{proof}
	The truth table of Game~\ref{game:3partydet} is given in Table~\ref{tab:3partydet}.
	\begin{table}
		\centering
		\begin{tabular}{ccc|c|ccc}
			$\hat A$&$\hat B$&$\hat C$&$\mathrm{maj}(\hat A,\hat B,\hat C)$&$X$&$Y$&$Z$\\
			\hline
			0&0&0&0&0&0&0\\
			0&0&1&0&1&0&0\\
			0&1&0&0&0&0&1\\
			0&1&1&1&0&0&1\\
			1&0&0&0&0&1&0\\
			1&0&1&1&1&0&0\\
			1&1&0&1&0&1&0\\
			1&1&1&1&0&0&0
		\end{tabular}
		\caption{Conditions for winning Game~\ref{game:3partydet}.}
		\label{tab:3partydet}
	\end{table}
	A necessary condition for predefined causal order is that at least one party is not in the causal future of any other party (see Lemma~\ref{lemma:necessarycondition}).
	Without loss of generality, let~$R$ be this party.
	Therefore, party~$R$ has to output~$X$ without learning any other random variable than~$\hat A$.
	By inspecting Table~\ref{tab:3partydet}, we see that~$X=0$ is~$R$'s best guess.
	Thus, in at least 2 out of the~8 cases, the game is lost, resulting in~$p_\mathrm{succ3}\leq 3/4$.
\end{proof}
\begin{theorem}
	{\bf\cite{Baumeler:2015wx}}
	\label{thm:3partydetperfect}
	By using the logically consistent classical process framework, Game~\ref{game:3partydet} can be won perfectly.
\end{theorem}
The logically consistent classical process~$E$ to perfectly win Game~\ref{game:3partydet} is
\begin{eqnarray*}
	P_{\bm I|\bm O}(\bm i,\bm o)=\cases{
		1&for $\mathrm{maj}(\bm o)=0\wedge i_R=o_T\wedge i_S=o_R\wedge i_T=o_S$\\
		1&for $\mathrm{maj}(\bm o)=1\wedge i_R=o_S\oplus 1\wedge i_S=o_T\oplus 1$\\
		&$\quad\wedge i_T=o_R\oplus 1$\\
		0&otherwise.
	}
\end{eqnarray*}
The local operations are
\begin{eqnarray*}
	\mathcal{E}(x,o_R,a,i_R)=\cases{
		1&for $x=i_R\wedge o_R=a$\\
		0&otherwise,
	}\\
	\mathcal{F}(y,o_S,b,i_S)=\cases{
		1&for $y=i_S\wedge o_S=b$\\
		0&otherwise,
	}\\
	\mathcal{G}(z,o_T,c,i_T)=\cases{
		1&for $z=i_T\wedge o_T=c$\\
		0&otherwise.
	}
\end{eqnarray*}
A graphical representation of the classical process is given in Figure~\ref{fig:detpro}.
\begin{figure}
	\centering
	\begin{tikzpicture}
		\def\r{1}
		\def\d{20}
		\def\dx{3}
		\def\hd{.5}

		\draw[->] (-\dx,0)++(90-\d:\r) arc (90-\d:-30+\d:\r);
		\draw[->] (-\dx,0)++(-30-\d:\r) arc (-30-\d:-150+\d:\r);
		\draw[->] (-\dx,0)++(-150-\d:\r) arc (-150-\d:-270+\d:\r);
		\draw (-\dx,0)++(90:\r) node (A) {$R$};
		\draw (-\dx,0)++(-30:\r) node (B) {$S$};
		\draw (-\dx,0)++(-150:\r) node (C) {$T$};

		\draw[<-] (\dx,0)++(90-\d:\r) arc (90-\d:-30+\d:\r);
		\draw[<-] (\dx,0)++(-30-\d:\r) arc (-30-\d:-150+\d:\r);
		\draw[<-] (\dx,0)++(-150-\d:\r) arc (-150-\d:-270+\d:\r);
		\draw (\dx,0)++(90:\r) node (A2) {$R$};
		\draw (\dx,0)++(-30:\r) node (B2) {$S$};
		\draw (\dx,0)++(-150:\r) node (C2) {$T$};
		\draw (\dx,0)++(90-60:\r+.3) node (n1) {$\oplus 1$};
		\draw (\dx,0)++(-30-60:\r+.3) node (n2) {$\oplus 1$};
		\draw (\dx,0)++(-150-60:\r+.3) node (n3) {$\oplus 1$};

		\draw (-\dx-\r-\hd,\r+0.5) node (half1) {$\mathrm{maj}(O_R,O_S,O_T)=0$};
		\draw (\dx-\r-\hd,\r+0.5) node (half2) {$\mathrm{maj}(O_R,O_S,O_T)=1$};
	\end{tikzpicture}
	\caption{Logically consistent classical process used to win Game~\ref{game:3partydet}.}
	\label{fig:detpro}
\end{figure}
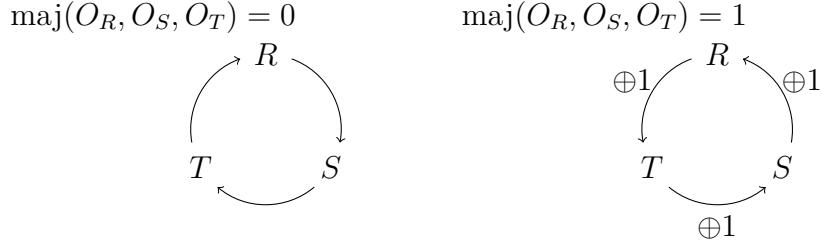

\subsection{Geometric representation}
The set of logically consistent classical processes forms a polytope~\cite{Baumeler:2015wx}, which is defined by the linear conditions~(\ref{eq:mnonneg}) and~(\ref{eq:trcond}).
Both classical processes used to perfectly win Games~\ref{game:3party} and~\ref{game:3partydet} are {\em extremal points\/} of the mentioned polytope.
In contrast to the process of Game~\ref{game:3partydet}, the classical process of Game~\ref{game:3party} is a {\em mixture\/} of {\em logically inconsistent\/} processes.
Such a classical process must be {\em fine-tuned}: Tiny variations of the probabilities make the classical process {\em logically inconsistent}.
This motivates the definition of a smaller polytope, where {\em all\/} extremal points represent {\em deterministic\/} classical processes.
This smaller polytope is called {\em deterministic-extrema polytope}.
Qualitative representations of both polytopes are given in Figure~\ref{fig:polytopes}.
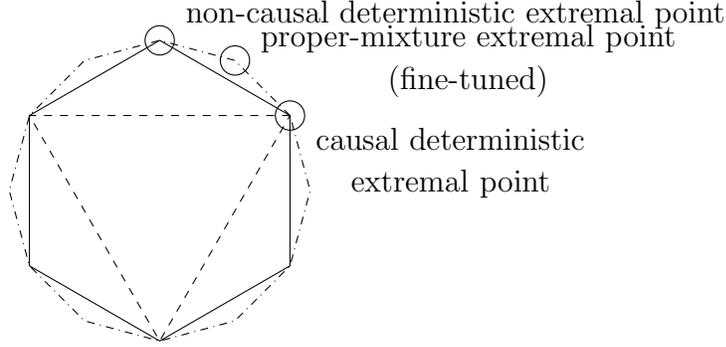
\begin{figure}
	\centering
	\begin{tikzpicture}[every text node part/.style={align=center}]
		\def\r{2}
		\draw[-,dashed] (30:\r) -- (150:\r) -- (270:\r) -- (30:\r);
		\draw[-] (30:\r) -- (90:\r) -- (150:\r) -- (210:\r) -- (270:\r) -- (330:\r) -- (30:\r);
		\draw[-,dashdotted] (30:\r) node[draw,solid,circle,label={-4:causal deterministic\\ extremal point}] (DC) {} -- (60:\r) node[draw,solid,circle,label={0:proper-mixture extremal point\\(fine-tuned)}] (PC) {} -- (90:\r) node[draw,solid,circle,label={3:non-causal deterministic extremal point}] (DNC) {} -- (120:\r) -- (150:\r) -- (180:\r) -- (210:\r) -- (240:\r) -- (270:\r) -- (300:\r) -- (330:\r) -- (0:\r) -- (30:\r);
	\end{tikzpicture}
	\caption{Qualitative representation of the polytopes discussed in the article. The outer polytope consists of all logically consistent classical processes. The middle polytope is the same as the outer except that all {\em probabilistic\/} extremal points are omitted. The innermost polytope is the polytope of all {\em causal\/} classical processes.}
	\label{fig:polytopes}
\end{figure}

\section{Characterizing logical consistency with fixed-points}
\label{sec:fixedpoint}
The above considerations on classical processes allow us to characterize logically consistent classical process via functions with {\em unique\/} fixed-points.
For that purpose, we redefine party~$R=(\mathcal{A},\mathcal{X},\bm{f})$ with~$\bm{f}:\mathcal{A}\times\mathcal{I}_R\rightarrow\mathcal{X}\times\mathcal{O}_R$ where~$\mathcal{A}$,~$\mathcal{X}$,~$\mathcal{I}_R$, and~$\mathcal{O}_R$ are sets.
Similarly, we redefine~$S=(\mathcal{B},\mathcal{X},\bm{g})$ with~$\bm{g}:\mathcal{B}\times\mathcal{I}_S\rightarrow\mathcal{Y}\times\mathcal{O}_S$ and~$T=(\mathcal{C},\mathcal{Z},\bm{h})$ with~$\bm{h}:\mathcal{C}\times\mathcal{I}_T\rightarrow\mathcal{Z}\times\mathcal{O}_T$.
Further, let~$f:\mathcal{I}_R\rightarrow\mathcal{O}_R$,~$g:\mathcal{I}_S\rightarrow\mathcal{O}_S$, and~$h:\mathcal{I}_T\rightarrow\mathcal{O}_T$ be functions.
\begin{theorem}[Deterministic extremal points modeled by functions with unique fixed-point]
	\label{thm:fp}
	A function~$\bm{e}:\mathcal{O}_R\times\mathcal{O}_S\times\mathcal{O}_T\rightarrow\mathcal{I}_R\times\mathcal{I}_S\times\mathcal{I}_T$ represents an extremal point of the deterministic-extrema polytope if and only if
	\begin{eqnarray*}
		\forall f,g,h,\exists ! (k,\ell,m):\,(k,\ell,m)=\bm{e}(f(k),g(\ell),h(m))
		\,,
	\end{eqnarray*}
	where~$\exists !$ is the uniqueness quantifier.
\end{theorem}
\begin{proof}
	Recall the total-probability condition~(\ref{eq:trcond}) that we express using stochastic matrices
	\begin{eqnarray*}
		\forall \mathcal{E'},\mathcal{F'},\mathcal{G'}\in\mathcal{D}:\,\mathrm{Tr}
		\left( \bm{E}\left(\bm{\mathcal{E'}}\otimes \bm{\mathcal{F'}}\otimes \bm{\mathcal{G'}}\right) \right)=1
		\,,
	\end{eqnarray*}
	where~$\mathcal{D}$ is the set of all {\em deterministic\/} local operations.
	By definition, a classical process~$E$ that is an extremal point of the deterministic-extrema polytope is deterministic.
	This implies that the matrix
	\begin{eqnarray*}
		\bm{M}=\bm{E}\left(\bm{\mathcal{E'}}\otimes \bm{\mathcal{F'}}\otimes \bm{\mathcal{G'}}\right)
		\,,
	\end{eqnarray*}
	for some~$\mathcal{E'},\mathcal{F'},\mathcal{G'}\in\mathcal{D}$, is a stochastic matrix with trivial probabilities.
	The diagonal~$\bm q$ of~$\bm M$ consists of~$0$'s and exactly one~$1$.
	In particular, by applying the stochastic matrix~$\bm M$ to~$\bm q$ results in~$\bm q$, {\em i.e.},~$\bm{M}\bm{q}=\bm{q}$, which means that~$\bm q$ is a unique fixed-point of~$\bm M$.
\end{proof}

Theorem~\ref{thm:fp} implies that deterministic causal and non-causal classical processes can be expressed as functions with one and only one fixed-point if the former is composed with arbitrary deterministic local operations.
Such a function for a causal classical process is given in Example~\ref{ex:fpcausal} and for a non-causal classical process in Example~\ref{ex:fpnoncausal}.
\begin{example}[Causal classical process as a function]
	\rm
	\label{ex:fpcausal}
	Let the sets~$\mathcal{A}$,~$\mathcal{X}$,~$\mathcal{I}_R$, and~$\mathcal{O}_R$, and likewise for the other two parties, be~$\{0,1\}$.
	Consider the causal classical process~$E$ with an identity channel from party~$R$ to~$S$ and an identity channel from party~$S$ to~$T$, and where~$R$ receives the element~$0$ from the environment.
	We express~$E$ as a function~$\bm e:\mathcal O_R\times\mathcal O_S\times\mathcal O_T\rightarrow\mathcal I_R\times\mathcal I_S\times\mathcal I_T$:
	\begin{eqnarray*}
		\bm e:(o,p,q)\mapsto(0,o,p)
		\,.
	\end{eqnarray*}
	Since the sets~$\mathcal O_R$ and~$\mathcal I_R$, and likewise for the other parties, are binary, the local operation of a party is one out of four functions:
	\begin{eqnarray*}
		d_\mathrm{id}&:&i\mapsto i\,,\\
		d_\mathrm{not}&:&i\mapsto i\oplus 1\,,\\
		d_0&:&i\mapsto 0\,,\\
		d_1&:&i\mapsto 1\,.
	\end{eqnarray*}
	For any choice~$f,g,h$ of the local operations, the function
	\begin{eqnarray*}
		\bm e\circ(f,g,h):\mathcal I_R\times\mathcal I_S\times\mathcal I_T\rightarrow\mathcal I_R\times\mathcal I_S\times\mathcal I_R
	\end{eqnarray*}
	has a unique fixed-point.
	We discuss a few cases.
	First, let~$f=g=h=d_\mathrm{id}$.
	Then, the fixed-point is~$(0,0,0)$, as can be verified by Table~\ref{tab:ex1}, which describes the composed function~$\bm e\circ(d_\mathrm{id},d_\mathrm{id},d_\mathrm{id})$.
	\begin{table}
		\centering
		\begin{tabular}{ccc|ccc}
			$\mathcal I_R$&$\mathcal I_S$&$\mathcal I_T$&$\mathcal I_R$&$\mathcal I_S$&$\mathcal I_T$\\
			\hline
			0&0&0&0&0&0\\
			0&0&1&0&0&0\\
			0&1&0&0&0&1\\
			0&1&1&0&0&1\\
			1&0&0&0&1&0\\
			1&0&1&0&1&0\\
			1&1&0&0&1&1\\
			1&1&1&0&1&1
		\end{tabular}
		\caption{Classical process from Example~\ref{ex:fpcausal} as a function composed with the identity as local operations of all three parties.}
		\label{tab:ex1}
	\end{table}
	In the case~$f=g=d_\mathrm{id}$ and~$h=d_\mathrm{not}$, the fixed-point is~$(0,0,0)$ as well.
	In the case~$f=g=h=d_\mathrm{not}$, the fixed-point is~$(0,1,0)$.
	A final case we express explicitly is~$f=g=h=d_1$; the fixed-point is~$(0,1,1)$.
\end{example}
\begin{example}[Non-causal classical process as a function]
	\rm
	\label{ex:fpnoncausal}
	We use the same parties as they are described in Example~\ref{ex:fpcausal}.
	Take the logically consistent classical process that can be used to win Game~\ref{game:3partydet}.
	A graphical representation of this classical process is given in Figure~\ref{fig:detpro}.
	We rewrite the process as a function~$\bm e:\mathcal O_R\times\mathcal O_S\times\mathcal O_T\rightarrow\mathcal I_R\times\mathcal I_S\times\mathcal I_T$:
	\begin{eqnarray*}
		\bm e:(o,p,q)\mapsto((p\oplus 1)q,o(q\oplus 1),(o\oplus 1)p)
		\,.
	\end{eqnarray*}
	If the parties use the identity as local operations, then~$\bm e\circ(d_\mathrm{id},d_\mathrm{id},d_\mathrm{id})$ has a unique fixed-point~$(0,0,0)$, as can be verified by inspecting Table~\ref{tab:ex2}.
	\begin{table}
		\centering
		\begin{tabular}{ccc|ccc}
			$\mathcal I_R$&$\mathcal I_S$&$\mathcal I_T$&$\mathcal I_R$&$\mathcal I_S$&$\mathcal I_T$\\
			\hline
			0&0&0&0&0&0\\
			0&0&1&1&0&0\\
			0&1&0&0&0&1\\
			0&1&1&0&0&1\\
			1&0&0&0&1&0\\
			1&0&1&1&0&0\\
			1&1&0&0&1&0\\
			1&1&1&0&0&0
		\end{tabular}
		\caption{Classical process from Example~\ref{ex:fpnoncausal} as a function composed with the identity as local operations of all three parties.}
		\label{tab:ex2}
	\end{table}
	For the local operations~\mbox{$f=g=h=d_\mathrm{not}$}, the unique fixed-point is~$(0,0,0)$ as well.
	In another case, if the local operations are~$f=g=d_\mathrm{id}$ and~$h=d_\mathrm{not}$, then the unique fixed-point is~$(1,0,0)$.
\end{example}

We present a statement similar to Theorem~\ref{thm:fp}, yet for any logically consistent classical process (probabilistic or deterministic).
\begin{theorem}[Logically consistent classical process and fixed-points]
	\label{thm:fpp}
	A conditional probability distribution~$E=P_{I_R,I_S,I_T|O_R,O_S,O_T}$ is a logically consistent classical process if and only if one of the decomposition~$E=\sum_i p_i D_i$, where~$\forall i:D_i$ is a deterministic distribution, has the property
	\begin{eqnarray*}
		\forall f,g,h:\,\sum_i p_i\left|\{(k,\ell,m)\,|\,(k,\ell,m)=\bm d_i(f(k),g(\ell),h(m))\}\right|=1
		\,,
	\end{eqnarray*}
	where~$\bm d_i:\mathcal O_R\times\mathcal O_S\times\mathcal O_T\rightarrow\mathcal I_R\times\mathcal I_S\times\mathcal I_T$ is a function modelling~$D_i$.
\end{theorem}
\begin{proof}
	The classical process~$\bm E$ can be decomposed as
	\begin{eqnarray*}
		\bm E=\sum_i p_i\bm D_i
		\,,
	\end{eqnarray*}
	where~$\forall i:\,D_i$ is a {\em deterministic\/} distribution.
	The total-probability condition~(\ref{eq:trcond}), that we express using stochastic matrices, is
	\begin{eqnarray*}
		\forall \mathcal{E'},\mathcal{F'},\mathcal{G'}\in\mathcal{D}:&\,&
		\mathrm{Tr}\left( \bm{E}\left(\bm{\mathcal{E'}}\otimes \bm{\mathcal{F'}}\otimes \bm{\mathcal{G'}}\right) \right)\\
		&=&\mathrm{Tr}\left( \sum_ip_i\bm D_i\left( \bm{\mathcal{E'}}\otimes\bm{\mathcal{F'}}\otimes\bm{\mathcal{G'}} \right) \right)\\
		&=&\sum_ip_i\mathrm{Tr}\left( \bm D_i\left( \bm{\mathcal{E'}}\otimes\bm{\mathcal{F'}}\otimes\bm{\mathcal{G'}} \right) \right)\\
		&=&\sum_ip_i\sum_{\bm j}\bm j^T\bm D_i\left( \bm{\mathcal{E'}}\otimes\bm{\mathcal{F'}}\otimes\bm{\mathcal{G'}} \right)\bm j\\
		&=&1
		\,,
	\end{eqnarray*}
	where~$\mathcal{D}$ is the set of all {\em deterministic\/} local operations, and where~$\bm j$ is a stochastic vector with all~$0$'s and exactly one~$1$.
	The expression
	\begin{eqnarray*}
		n_i:=\sum_{\bm j}\bm j^T\bm D_i\left( \bm{\mathcal{E'}}\otimes\bm{\mathcal{F'}}\otimes\bm{\mathcal{G'}} \right)\bm j
	\end{eqnarray*}
	is thus the number of fixed-points of the stochastic matrix~$\bm D_i\left( \bm{\mathcal{E'}}\otimes\bm{\mathcal{F'}}\otimes\bm{\mathcal{G'}} \right)$.
	From this we get~$\sum_i p_in_i=1$.
\end{proof}
The informal statement of Theorem~\ref{thm:fpp} is: The average number of fixed-points of a logically consistent classical process is one.
Clearly, the classical processes of Examples~\ref{ex:fpcausal} and~\ref{ex:fpnoncausal} can be rewritten as a convex combination of deterministic processes fulfilling this property.
However, there are classical processes which lie {\em outside\/} the deterministic-extrema polytope --- Theorem~\ref{thm:fp} does not apply to those --- to which we can apply Theorem~\ref{thm:fpp}.
Theorems~\ref{thm:fp} and~\ref{thm:fpp} can naturally be extend to more than three parties.

\begin{example}[Non-causal classical process as a mixture of functions]
	\rm
	\label{ex:finetuned}
	The classical process used to win Game~\ref{game:3party} is
	\begin{eqnarray*}
		P_{\bm I|\bm O}(\bm i,\bm o)=\cases{
			1/2&for $i_R=o_T\wedge i_S=o_R\wedge i_T=o_S$\\
			1/2&for $i_R=o_T\oplus 1\wedge i_S=o_R\oplus 1\wedge i_T=o_S\oplus 1$\\
			0&otherwise.
		}
	\end{eqnarray*}
	A graphical representation of this classical process is given in Figure~\ref{fig:stopro}.
	The classical process can be written as a convex combination~$E=(E_0+E_1)/2$ with
	\begin{eqnarray*}
		E_0(i_R,i_S,i_T,o_R,o_S,o_T)&=&(i_R=o_T\wedge i_S=o_R\wedge i_T=o_S)\,,\\
		E_1(i_R,i_S,i_T,o_R,o_S,o_T)&=&(i_R=o_T\oplus 1\wedge i_S=o_R\oplus 1\wedge i_T=o_S\oplus 1)\,.
	\end{eqnarray*}
	Now, we switch to the parties as defined in Example~\ref{ex:fpcausal} and rewrite~$E_0$ and~$E_1$ as
	functions~$\bm e_0:\mathcal O_R\times\mathcal O_S\times\mathcal O_T\rightarrow\mathcal I_R\times\mathcal I_S\times\mathcal I_T$ 
	and~$\bm e_1:\mathcal O_R\times\mathcal O_S\times\mathcal O_T\rightarrow\mathcal I_R\times\mathcal I_S\times\mathcal I_T$:
	\begin{eqnarray*}
		\bm e_0:(o,p,q)\rightarrow (q,o,p)\,,\\
		\bm e_1:(o,p,q)\rightarrow (q\oplus 1,o\oplus 1,p\oplus 1)\,.
	\end{eqnarray*}
	For the choice~$f=g=h=d_\mathrm{id}$ of local operations, the function~$\bm e_0\circ(d_\mathrm{id},d_\mathrm{id},d_\mathrm{id})$ has two fixed-points, and the function~$\bm e_1\circ(d_\mathrm{id},d_\mathrm{id},d_\mathrm{id})$ has no fixed-points (see Table~\ref{tab:ex3}).
	\begin{table}
		\centering
		\begin{tabular}{ccc|ccc}
			$\mathcal I_R$&$\mathcal I_S$&$\mathcal I_T$&$\mathcal I_R$&$\mathcal I_S$&$\mathcal I_T$\\
			\hline
			0&0&0&0&0&0\\
			0&0&1&1&0&0\\
			0&1&0&0&0&1\\
			0&1&1&1&0&1\\
			1&0&0&0&1&0\\
			1&0&1&1&1&0\\
			1&1&0&0&1&1\\
			1&1&1&1&1&1
		\end{tabular}
		\qquad
		\begin{tabular}{ccc|ccc}
			$\mathcal I_R$&$\mathcal I_S$&$\mathcal I_T$&$\mathcal I_R$&$\mathcal I_S$&$\mathcal I_T$\\
			\hline
			0&0&0&1&1&1\\
			0&0&1&0&1&1\\
			0&1&0&1&1&0\\
			0&1&1&0&1&0\\
			1&0&0&1&0&1\\
			1&0&1&0&0&1\\
			1&1&0&1&0&0\\
			1&1&1&0&0&0
		\end{tabular}
		\caption{The pair of functions that yield the classical process from Example~\ref{ex:finetuned} when mixed, as functions composed with the identity as local operations. The left table describes the function~$\bm e_0\circ(d_\mathrm{id},d_\mathrm{id},d_\mathrm{id})$, and the right table describes the function~$\bm e_1\circ(d_\mathrm{id},d_\mathrm{id},d_\mathrm{id})$.}
		\label{tab:ex3}
	\end{table}
	Thus, the average number of fixed-points is~$(2+0)/2=1$.
	An alternative choice of local operations is~$f=g=d_\mathrm{id}$ and~$h=d_\mathrm{not}$.
	Then~$\bm e_0\circ(d_\mathrm{id},d_\mathrm{id},d_\mathrm{not})$ has~$0$, and~$\bm e_1\circ(d_\mathrm{id},d_\mathrm{id},d_\mathrm{not})$ has~$2$ fixed-points:~$(0,1,0)$ and~$(1,0,1)$.
	As a last choice, consider~$f=g=h=d_0$, in which both composed functions have~$1$ fixed-point:~$(0,0,0)$ for the former function, and~$(1,1,1)$ for the latter function.
\end{example}

\section{A non-causal model for computation}
\label{sec:circuit}
Above we took the {\em operational approach}, where we think in terms of parties that can choose to perform arbitrary experiments within their laboratories.
A model for computation can be designed if we depart from this point of view and rather think in terms of {\em circuits}.
We briefly sketch such a model~\cite{inprep2}.
In the logical-consistency conditions of the process-matrix framework and its classical analogue, we consider {\em any local operation\/} of the parties.
In a circuit, however, there are no parties.
Thus, the logical-consistency condition has to be adopted.
Let~$\mathcal{C}$ be a classical circuit that consists of wires and gates.
\begin{definition}[Logically consistent circuit]
	\rm
	A classical circuit~$\mathcal{C}$ is called {\em logically consistent\/} if and only if there exists a probability distribution over all values on the wires where the probabilities are linear in the choice of the inputs.
\end{definition}
This definition allows gates to be connected in a circular way --- which is not allowed in the standard circuit model.
The wires which are connected to a gate on {\em one\/} side only are either the inputs or the outputs of the circuit.
Intuitively, if a closed circuit (when ignoring the open wires) admits a unique fixed-point, then it is logically consistent.
The wires, in that case, carry the value of the fixed-point.
This non-causal model for computation is more powerful than the standard circuit model.
As a demonstration, consider the following task.
\begin{task}[Fixed-point search]
	\rm
	\label{task:fp}
	Suppose you are given a black-box~$\mathrm{B}$ with the {\em promise\/} that~$\mathrm{B}$ has {\em exactly one\/} fixed-point.
	The task is to find the fixed-point with as few queries to~$\mathrm{B}$ as possible.
\end{task}
In the worst case, the minimal number of queries to~$\mathrm{B}$ in the standard circuit model is~$n-1$, where~$n$ is the input size.
\begin{theorem}[Task~\ref{task:fp} can be solved with one query to~$\mathrm{B}$]
	By using this new model of computation, a single query to black box~$\mathrm{B}$ suffices to solve Task~\ref{task:fp}.
\end{theorem}
\begin{proof}
	To solve the task, we use the logically consistent non-causal circuit given by Figure~\ref{fig:circuit} with~$a=0$.
	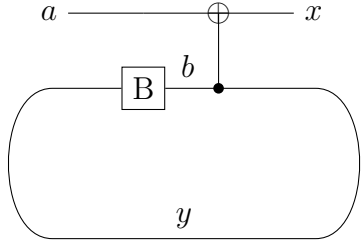
\begin{figure}
		\centering
		\begin{tikzpicture}
			\node[draw,shape=rectangle,minimum width=0.5cm,minimum height=0.5cm] (B) {$\mathrm{B}$};
			\draw[-] (B.east) -- ++(0.6,0) node[midway,above] {$b$} -- ++(1.4,0) to[out=0,in=0] ++(0,-2) -- ++(-3.5,0) node[midway,above] {$y$} to[out=180,in=180] ++(0,2) -- (B.west);
			\draw[-] (B.center)++(0,1) -- ++(-1,0) node[left] {$a$};
			\draw[-] (B.center)++(0,1) -- ++(1,0) node (O) {$\oplus$} -- ++(1,0) node[right] {$x$};
			\draw[-] (O.center) -- (O|-B);
			\fill (O|-B) circle [radius=2pt];
		\end{tikzpicture}
		\caption{Circuit to find the fixed-point~$x$ of~$\mathrm{B}$.}
		\label{fig:circuit}
	\end{figure}
	The stochastic matrix of the~$\mathrm{CNOT}$ gate~$C$ is
	\begin{eqnarray*}
		\bm{C}=\sum_{i,j=1}^n ((\bm i\oplus \bm j)\otimes\bm j)^T(\bm i\otimes\bm j)
		\,,
	\end{eqnarray*}
	and the stochastic matrix of the black-box~$\mathrm{B}$ is
	\begin{eqnarray*}
		\bm{B}=\sum_{i=1}^n \bm e_i^T\bm i\,,
	\end{eqnarray*}
	with
	\begin{eqnarray*}
		\left|\{i\,|\,i=e_i\}\right|=1
		\,.
	\end{eqnarray*}
	Since~$\mathrm{B}$ is a black-box, the values~$e_i$ are unknown to us.
	We write the probability of getting~$(b,x,y)$ conditioned on the input~$a$ with stochastic matrices and stochastic vectors:
	\begin{eqnarray*}
		P_{B,X,Y|A}(b,x,y,a)=
		\left(
		(\bm x\otimes \bm y)^T
		\bm{C}
		(\bm a\otimes \bm b)
		\right)
		\left(\bm b^T\bm{B}\bm y\right)
		\,,
	\end{eqnarray*}
	where~$\bm x$ is the stochastic vector of the value~$x$, as described on page~\pageref{p:sv}, and where~$\bm b^T\bm{B}\bm y$ is the probability of obtaining~$b$ when~$y$ is given as input to~$\mathrm{B}$.
	Now, if we use~$a=0$, we find the fixed-point~$x$:
	\begin{eqnarray*}
		P_{B,X,Y|A}(b,x,y,0)=
		\cases{
			1&for $x=b=y$\\
			0&otherwise.
		}
	\end{eqnarray*}
\end{proof}

So far we ignore how to deal with black boxes that have an {\em unknown number of fixed-points}, and whether, in that case, a computational advantage is achievable at all.

\section{Open questions}
The main open question in this field of research is whether nature violates causal inequalities.
Contrary to some superpositions of causal orders~\cite{Procopio:2015iw}, no experiment has been proposed or carried out to violate such inequalities.
Even though some results~\cite{Baumeler:2013wy,Brukner:2014vo} indicate a strong connection between non-causal correlations and non-local correlations --- the maximally achievable violation of a two-party causal inequality in the quantum setting is~$(2+\sqrt 2)/4$, which is the same value as Cirel'son's bound~\cite{Cirelson:1980fp}, and for three parties, the algebraic maximum is achievable~\cite{Greenberger:1989vx,Greenberger:1990it} ---, the connection remains unknown.
A violation of a two-party causal inequality up to the algebraic maximum by the use of higher-dimensional systems is also unknown so far; such a result would strengthen the connection between both types of correlations and could give insight on how these correlations are connected.
To our knowledge, it is also unknown which process matrices can be purified.
Since purification is of great importance in quantum information~\cite{Chiribella:2012kz}, purification of process matrices is also desired.
It has been suspected that the process matrices that {\em can\/} be purified {\em cannot\/} violate causal inequalities.
However, a counterexample has been provided recently~\cite{Baum,inprep}.
The same was conjectured for classical processes.
This has been refuted as well: For the classical case, it is known that classical processes from the deterministic-extrema polytope can be made reversible~\cite{Baum,inprep} --- this fact is the classical analogue of purification.
Of great interest are also information-theoretic considerations of non-causal correlations, similar to~\cite{Pawiowski:2009dt,Ibnouhsein:2015cf}.
For instance, it might be possible to quantify the amount of ``superposition'' versus the amount of ``loops'' a process matrix carries.
Finally, little is known on the computational power of process matrices, classical processes, and of our model of computation.
While the D-CTC model~\cite{Deutsch:1991jo} can efficiently solve problems that are in PSPACE~\cite{Aaronson:2009dy}, it is unknown what the complexity class is of the efficiently solvable problems in our model of computation.
\section*{Acknowledgements}
	We thank Mateus Ara{\'u}jo, Veronika Baumann, {\v C}aslav Brukner, Fabio Costa, Adrien Feix, Christina Giarmatzi, and Thomas Strahm for helpful discussions.
	Furthermore, we thank the anonymous referees for their detailed comments.
	The present work was supported by the Swiss National Science Foundation (SNF) and the National Centre of Competence in Research ``Quantum Science and Technology'' (QSIT).

\section*{References}
\bibliography{refs}

\end{document}